\documentclass[10pt,final,journal,letterpaper,oneside,twocolumn]{IEEEtran}
\usepackage[pdftex]{graphicx}
\usepackage[ruled]{algorithm2e}
\usepackage{algorithmic}
\usepackage[small]{caption}
\usepackage{float}
\usepackage{xspace}
\usepackage{fancyhdr}
\usepackage{amssymb,amsmath}
\usepackage{subfig}
\usepackage{epstopdf}
\usepackage{bbding}
\usepackage{cite} 
\usepackage{color}
 \usepackage[english]{babel}
\usepackage{amsthm}

\newtheorem{lemma}{Lemma} 

\begin{document}
%
\title{Beyond Diagonal RIS Enhanced Cognitive Radio Enabled Multilayer Non-Terrestrial Networks}

\author{Wali Ullah Khan, Chandan Kumar Sheemar, Eva Lagunas, Symeon Chatzinotas  \thanks{Authors are with the Interdisciplinary Centre for Security, Reliability, and Trust (SnT), University of Luxembourg, 1855 Luxembourg City, Luxembourg (emails: \{waliullah.khan, chandankumar.sheemar, eva.lagunas, symeon.chatzinotas\}@uni.lu).

}}%

\markboth{IEEE Conference
}
{Shell \MakeLowercase{\textit{et al.}}: Bare Demo of IEEEtran.cls for IEEE Journals} 

\maketitle

\begin{abstract}
Beyond diagonal reconfigurable intelligent surfaces (BD-RIS) have emerged as a transformative technology for enhancing wireless communication by intelligently manipulating the propagation environment. Its interconnected elements offer enhanced control over signal redirection, making it a promising solution for integrated terrestrial and non-terrestrial networks (NTNs). This paper explores the potential of BD-RIS in improving cognitive radio enabled multilayer non-terrestrial networks. We formulate a joint optimization problem that maximizes the achievable spectral efficiency by optimizing BD-RIS phase shifts and secondary transmitter power allocation while controlling the interference temperature from the secondary network to the primary network. To solve this problem efficiently, we decouple the original problem and propose a novel solution based on an alternating optimization approach. Simulation results demonstrate the effectiveness of BD-RIS in cognitive radio enabled multilayer NTNs.
\end{abstract}


\IEEEpeerreviewmaketitle

\section{Introduction}
The integration of non-terrestrial networks (NTNs) in 6G aims to provide seamless global coverage, ubiquitous connectivity, ultra-reliable communications, and enhanced network access in remote or underserved areas \cite{10716670}. Unlike traditional terrestrial networks, NTNs leverage satellites, high-altitude platform stations (HAPS), and unmanned aerial vehicles (UAVs) to extend network access beyond conventional limits \cite{sheemar2025joint,iacovelli2024holographic}. These multi-layer architectures enable robust wireless communication across diverse geographical landscapes, supporting emerging applications such as Internet of Things (IoT), vehicular networks, and disaster recovery systems \cite{10097680}. However, despite their vast potential, NTNs face several challenges, including spectrum scarcity, high propagation losses, dynamic channel conditions, and energy efficiency due to long-distance between transmitter and receiver.

To address these challenges, cognitive radio networks and beyond diagonal reconfigurable intelligent surfaces (BD-RIS) have emerged as key enablers for enhancing spectrum efficiency and improving network adaptability \cite{9514409}. Cognitive radio networks operate with a primary network and an underlaid secondary network. The users of the secondary network opportunistically access the spectrum of the primary network without endangering the service quality of its associated users \cite{10013700}. This dynamic spectrum access mechanism helps mitigate spectrum scarcity issues in NTNs, enabling more efficient and adaptive communication. Meanwhile, BD-RIS offers an advanced means of beamforming and interference suppression by reconfiguring the wireless environment through both diagonal and non-diagonal entries of the phase shift matrix \cite{9913356}. Using cognitive radio and BD-RIS, 6G NTNs can achieve higher spectral efficiency, improved energy efficiency, and more reliable communications in challenging environments.

Most of the existing literature considered BD-RIS in terrestrial networks. The works \cite{10587164,de2024semi} explored channel estimation approaches in BD-RIS assisted multiple-input multiple-output (MIMO) systems. The papers \cite{10834443,li2024beyond,10787237 } used BD-RIS to maximize the achievable rate/ spectral efficiency and energy efficiency of wireless systems. Moreover, \cite{10777522,10693959} investigated the performance of BD-RIS in integrated sensing and communication-based wireless systems. Furthermore, \cite{10817282,10817342} utilized BD-RIS in terahertz communication to extend the wireless coverage in indoor and outdoor environments.

Recently, few works have explored the potential of BD-RIS in NTNs. For instance, the work \cite{10716670} utilized BD-RIS to assist the communication of the LEO satellite. They optimized power allocation and phase shift design to maximize the spectral efficiency of the system. Another work \cite{khan2024integration} used a transmissive BD-RIS mounted UAV communication to support multiuser multicarrier communication. The authors optimized transmissive beamforming and power allocation to maximize the achievable data rate of the system. Similarly, \cite{khan2025transmissive} adopted transmissive BD-RIS mounted LEO communication to support IoT devices on the ground. They maximized the sum rate of the system through joint power allocation and phase shift design. Considering the available literature, research on BD-RIS in NTNs is still in its early stages and requires further exploration.

This work considers a new framework that considers cognitive radio-enabled multilayer NTNs, which consists of a primary LEO network and BD-RIS mounted secondary HAPS network. In particular, the proposed framework seeks to optimize BD-RIS phase shift and HAPS transmit power in the secondary network while ensuring the service quality of the primary LEO network by imposing the interference temperature constraint. To handle the non-convexity of the spectral efficiency problem, we employ alternating optimization, where power allocation is optimized using water-filling and KKT conditions, while phase shift is designed via Riemannian optimization. Numerical results are compared with conventional diagonal RIS configuration.

\emph{Paper Organization:} The rest of the paper is organized as follows. Section  \ref{sez_2} presents the system model and the problem formulation. Section \ref{sez_3} presents the solution. Finally, Section \ref{sez_4} and \ref{sez_5} present the simulation results and conclusions, respectively.

\begin{figure}[!t]
\centering
\includegraphics [width=.48\textwidth]{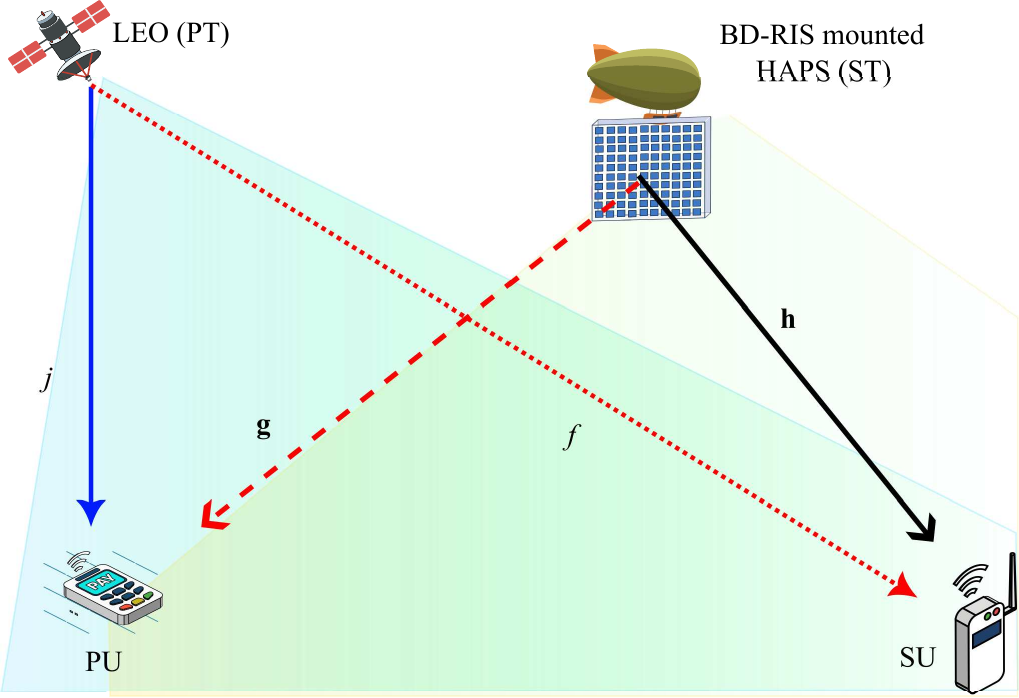}
\caption{System model of BD-RIS enhanced cognitive radio network.}
\label{CRNsm}
\end{figure}

\section{System Model and Problem Formulation} \label{sez_2}
Let us consider a cognitive radio-enabled multilayer NTN that consists of a primary network with the primary LEO transmitter (PT) serving a primary user (PU) and a secondary transmissive BD-RIS mounted HAPS transmitter (ST) serving a secondary user (SU) as shown in Fig. \ref{CRNsm}\footnote{Although we consider a single user in both primary and secondary networks, our framework can be extended to a multicarrier-multiuser scenario. For instance, future work could allocate one user per carrier, making the number of users equal to the number of carriers. All carriers would be exclusively reused in the secondary network, forming a cognitive network and improving spectral efficiency.}. BD-RIS consists of \(M\) reconfigurable elements and follows a fully-connected architecture with interconnected elements. The primary and secondary networks share the same spectrum to maximize spectral efficiency. Thus, the objective is to maximize the achievable spectral efficiency of the secondary HAPS network by optimizing the transmit power of the ST and the scattering matrix of the BD-RIS, subject to the interference temperature at the PU. We consider the HAPS to be in a nearly stationary position while hovering \cite{9515574}. Therefore, the channel vector $\mathbf{h}$ from ST to SU can be modeled as Rician fading such as:
\begin{equation}
\mathbf{h} = \sqrt{\frac{\hat{h}}{d^2}}\Bigg (\sqrt{\frac{K}{K+1}} \mathbf{h}_{\text{LoS}} + \sqrt{\frac{1}{K+1}} \mathbf{h}_{\text{NLoS}}\Bigg)\in \mathbb{C}^{M \times 1},\label{1}
\end{equation}
where \(K\) is the Rician factor, $\hat{h}$ denotes the channel power gain, $d$ shows the distance between ST and SU, \(\mathbf{h}_{\text{LoS}}\) represents the deterministic line-of-sight (LoS) component, and \(\mathbf{h}_{\text{NLoS}}\) follows a Rayleigh distribution. According to \cite{10839492}, the LoS component \(\mathbf{h}_{\text{LoS}}\) can be described as:
\begin{align}
&\mathbf{h}_{\text{LoS}}=[1,e^{-j\delta\sin{\theta}\cos{\varphi}},\dots,e^{-j\delta\sin{\theta}\cos{\varphi}(M_x-1)}]^T\nonumber\\& \otimes[1,e^{-j\delta\sin{\theta}\sin{\varphi}},\dots,e^{-j\delta\sin{\theta}\sin{\varphi}(M_y-1)}]^T
\end{align}
where \(\delta=(2\pi f_c q)/c\) with $q$ being the distance between adjacent phas shift elements, $f_c$ is the carrier frequency and $c$ shows the speed of light. Subsequently, the channel from ST to PU and PT to SU can be denoted as $\mathbf{g}\in \mathbb{C}^{M \times 1}$ and $f\in \mathbb{C}^{1 \times 1}$. Please note that $\mathbf{g}$ and $f$ follow Rician fading and are assumed to be modeled similarly as in (\ref{1}).

Given the aforementioned notations, we can express the received signal at the SU as

\begin{equation}
    y = \mathbf{h} \mathbf{\Phi}\sqrt{P_s} x +f\sqrt{Q_p}z+ n_{s},
\end{equation}
where \(x\) is the i.i.d unit-variance data stream transmitted from the ST, $z$ is the i.i.d unit-variance data stream transmitted from the PT, $\boldsymbol{\Phi}\in\mathbb{C}^{M_x\times M_y}$ denotes the phase shift matrix of transmissive BD-RIS such that $\boldsymbol{\Phi}\boldsymbol{\Phi}^H={\bf I}_M$, and $M_x$ and $M_y$ denote the number of elements along the x-axis )(rows) and y-axis (columns), respectively, and \(n_{s} \sim \mathcal{CN}(0, \sigma^2)\) is the additive white Gaussian noise (AWGN). Furthermore, $P_s$ and $Q_p$ represent the transmit power of ST and PT, respectively.
The achievable rate at the SU based on the received signal can be stated as
\begin{equation}
    R_{s} = \log_2 \left( 1 + \gamma_s \right),
\end{equation}
where $\gamma_s$ is the signal-to-interference plus noise ratio (SINR), given as
\begin{align}
\gamma_s=\frac{|\mathbf{h} \mathbf{\Phi}|^2 P_{s} }{\sigma^2+|f|^2Q_p}.\label{4}
\end{align}
To ensure the service quality of the primary network, the interference power at the PU is constrained by the interference temperature limit \(I_{\text{th}}\) with the following constraint
\begin{equation}
 |\mathbf{g} \mathbf{\Phi}|^2 P_{s} \leq I_{\text{th}}.
\end{equation}

The objective of this work is to maximize the achievable spectral efficiency of the secondary network, while controlling the interference temperature at the PU. Formally, the optimization problem for such purpose can be stated as
\begin{align}
 (\mathcal{P}) \quad   \max_{P_{s}, \mathbf{\Phi}} \quad & \log_2 \left( 1 + \gamma_s \right) \label{obj}\\
    \text{s.t.} \quad C_1:\ & |\mathbf{g} \mathbf{\Phi}|^2 P_{s} \leq I_{\text{th}}, \label{C1}\\
    C_2:\ & 0 \leq P_{s} \leq P_{\max},\label{C2} \\
   C_3:\ & \mathbf{\Phi}\mathbf{\Phi}^H = \textbf{I}_m,\label{C3}
\end{align}
where $C_1$ ensures the service quality of PU by controlling the interference temperature coming from ST, $C_2$ limits the transmit power of ST, and $C_3$ denotes the BD-RIS constraint on its phase-shift response. 

It is noteworthy that this is a highly non-convex optimization problem due to the multiplicative coupling between $P_s$ and $\boldsymbol{\Phi}$. Therefore, finding its global optimum is extremely challenging. Furthermore, as similar problems have not yet been explored in the literature, this drives the quest for novel algorithm design highly tailored for this scenario.


\section{Proposed Solution}  \label{sez_3}
To solve the optimization problem in (\ref{obj}), we resort to an alternating optimization framework. Namely,  we first decouple the problem into two manageable subproblems, i.e., transmit power optimization for ST, and phase shift design for BD-RIS. To solve the first problem, the wate-filling method is tailored to this challenging problem to maximize the objective function. Subsequently,  semi-definite relaxation for phase shift design for BD-RIS is adopted. 

\subsection{Optimal Power allocation}
In this part, we assume the phase-shift response of the BD-RIS to be fixed and consider finding the optimal power allocation. For such a purpose, the original optimization problem can be  expressed as
\begin{align} \label{ref_power}
 (\mathcal{P}_1)   \max_{P_{s}} \quad & \log_2 \left( 1 + \gamma_s \right) \\
    \text{s.t.} \quad C_1:\ & |\mathbf{g} \mathbf{\Phi}|^2 P_{s} \leq I_{\text{th}}\\
   C_2:\ &  0 \leq P_{s} \leq P_{\max}.
\end{align}
Note that finding the optimal \(P_s\) is not an easy task. To proceed, we first state the following result.

\begin{lemma}
    Given the BD-RIS response fixed from previous iteration, the objective function \eqref{ref_power} is concave with respect to the power $P_{s}$.
\end{lemma}

\begin{proof}
    To prove this result we proceed as follows. Note that the objective function can be recalled here again as:
\begin{equation}
    f(P_s) = \log_2 \left( 1 + \gamma_s \right),
\end{equation}
where the SINR is defined as
\begin{equation}
    \gamma_s = \frac{|\mathbf{h} \mathbf{\Phi}|^2 P_s}{\sigma^2 + |f|^2 Q_p}.
\end{equation}
Substituting \(\gamma_s\) into \(f(P_s)\), we obtain
\begin{equation}
    f(P_s) = \log_2 \left( 1 + \frac{|\mathbf{h} \mathbf{\Phi}|^2 P_s}{\sigma^2 + |f|^2 Q_p} \right).
\end{equation}

To prove concavity, we compute the first and second derivatives of \(f(P_s)\). Using the chain rule, the first derivative can be expressed as:
\begin{align}
    \frac{df}{dP_s} &= \frac{1}{\ln(2)} \cdot \frac{|\mathbf{h} \mathbf{\Phi}|^2}{1 + \frac{|\mathbf{h} \mathbf{\Phi}|^2 P_s}{\sigma^2 + |f|^2 Q_p}} \cdot \frac{1}{\sigma^2 + |f|^2 Q_p} \nonumber \\
    &= \frac{|\mathbf{h} \mathbf{\Phi}|^2}{\ln(2) \left( |\mathbf{h} \mathbf{\Phi}|^2 P_s + \sigma^2 + |f|^2 Q_p \right)}.
\end{align}
Taking the derivative again, we obtain:
\begin{align}
    \frac{d^2 f}{dP_s^2} &= -\frac{|\mathbf{h} \mathbf{\Phi}|^4}{\ln(2) \left( |\mathbf{h} \mathbf{\Phi}|^2 P_s + \sigma^2 + |f|^2 Q_p \right)^2}.
\end{align}
It can be observed that \( |\mathbf{h} \mathbf{\Phi}|^4 \), \( \ln(2) \), and \( \left( |\mathbf{h} \mathbf{\Phi}|^2 P_s + \sigma^2 + |f|^2 Q_p \right)^2 \) are all positive, it follows that:
\begin{equation}
    \frac{d^2 f}{dP_s^2} \leq 0, \quad \forall P_s \geq 0.
\end{equation}
Since the second derivative is non-positive for all \( P_s \geq 0 \), thus the function \( f(P_s) = \log_2(1 + \gamma_s) \) is concave in \( P_s \).
\end{proof}
Since the objective function is concave in \(P_s\), we can find the optimal power allocation based on the water-filling principle. Using the effective SINR in (\ref{4}), we take the derivative of the objective function and set the gradient to zero. This leads to the following solution satisfying the Karush-Kuhn-Tucker (KKT) condition

\begin{equation}  
    P_s^* = \min \left( \left[ \frac{1}{\lambda} - \frac{\sigma^2+|f|^2 Q_p}{|\mathbf{h} \mathbf{\Phi}|^2} \right]^+, P_{\max} \right).\label{p}
\end{equation}
A mathematical proof demonstrating this result is provided in Appendix A. In \eqref{p}, \(\lambda\) denotes the Lagrange multiplier associated with the interference constraint and its optimal value is given as
\begin{equation}
    \lambda = \frac{|\mathbf{g} \mathbf{\Phi}|^2}{I_{\text{th}}}.
\end{equation}
 Thus, the optimal \(P_s^*\) is determined by substituting \(\lambda\) and ensuring the interference constraint holds.
\subsection{BD-RIS Optimization}
Subsequently, given the \(P^*_s\), we optimize the BD-RIS phase shift matrix \(\mathbf{\Phi}\). This optimization problem can be formulated as

\begin{align}
 (\mathcal{P}_2) \quad  \max_{\mathbf{\Phi}} \quad & \log_2 \left( 1 + \frac{|\mathbf{h} \mathbf{\Phi}|^2 P_{s} }{\sigma^2+|f|^2Q_p} \right) \\
    \text{s.t.} \quad & |\mathbf{g} \mathbf{\Phi}|^2 P_{s} \leq I_{\text{th}}, \\
    & \mathbf{\Phi} \mathbf{\Phi}^H = \mathbf{I}_M.
\end{align}
We can see that the unitary constraint \(\mathbf{\Phi} \mathbf{\Phi}^H = \mathbf{I}_M\) in $(\mathcal{P}_2)$ is non-convex, we employ Riemannian optimization to iteratively update \(\mathbf{\Phi}\) on the Stiefel manifold.
To ensure feasibility, we introduce the Lagrange multiplier \(\mu\) for the interference constraint as:
\begin{equation}
    \mathcal{L}(\mathbf{\Phi}, \mu) = f(\mathbf{\Phi}) - \mu \left(|\mathbf{g} \mathbf{\Phi}|^2 P_{s} - I_{\text{th}}\right).
\end{equation}
where $f(\mathbf{\Phi})=\frac{|\mathbf{h} \mathbf{\Phi}|^2 P_{s} }{\sigma^2+|f|^2Q_p}$. Since \(\mathbf{\Phi}\) lies on the Stiefel manifold, we project the Euclidean gradient onto the tangent space. The Euclidean gradient of \(\mathcal{L}(\mathbf{\Phi}, \mu)\) can be expressed as:
\begin{equation}
    \nabla_{\mathbf{\Phi}} \mathcal{L} = \frac{2 P_s}{\ln(2)} \mathbf{h}^H \mathbf{h} \mathbf{\Phi} - 2 \mu P_s \mathbf{g}^H \mathbf{g} \mathbf{\Phi}.
\end{equation}
The Euclidean gradient is used to compute the Riemannian gradient, which lies on the tangent space of the Stiefel manifold. Next, the Riemannian gradient is obtained using the following projection such as:
\begin{equation}
    \text{Proj}_{\mathbf{\Phi}}(\nabla_{\mathbf{\Phi}} \mathcal{L}) = \nabla_{\mathbf{\Phi}} \mathcal{L} - \mathbf{\Phi} \left(\mathbf{\Phi}^H \nabla_{\mathbf{\Phi}} \mathcal{L} + \nabla_{\mathbf{\Phi}} \mathcal{L}^H \mathbf{\Phi} \right)/2.
\end{equation}
This projection ensures that the gradient update respects the unitary constraint $\mathbf{\Phi} \mathbf{\Phi}^H = \mathbf{I}_M$. The term $\mathbf{\Phi} \left(\mathbf{\Phi}^H \nabla_{\mathbf{\Phi}} \mathcal{L} + \nabla_{\mathbf{\Phi}} \mathcal{L}^H \mathbf{\Phi} \right)/2$ removes the component of the gradient that points outside the tangent space. Note that Riemannian gradient is used to update $\boldsymbol{\Phi}$ iteratively.
Next, we use a step size \(\eta_k\) to updated the phase shift matrix of BD-RIS as:
\begin{equation}
    \mathbf{\Phi}_{k+1} = \mathbf{\Phi}_k - \eta_k \cdot \text{Proj}_{\mathbf{\Phi}}(\nabla_{\mathbf{\Phi}} \mathcal{L}).
\end{equation}
After the update, $\boldsymbol{\phi}_{k+1}$ may no longer satisfy the unitary constraint, so we enforce it using singular value decomposition (SVD). It can be expressed as:
\begin{equation}
    \mathbf{\Phi} = \mathbf{U} \mathbf{V}^H, \quad \text{where } \mathbf{U} \mathbf{\Sigma} \mathbf{V}^H = \text{SVD}(\mathbf{\Phi}).
\end{equation}
The SVD of $\boldsymbol{\phi}_{k+1}$ decomposes it into $\mathbf{U} \mathbf{\Sigma} \mathbf{V}^H$, where $\mathbf{U}$ and $\mathbf{V}$ are unitary matrices, and $\mathbf{\Sigma}$ is a diagonal matrix of singular values. By setting $ \mathbf{\Phi} = \mathbf{U} \mathbf{V}^H$, we ensure that $\boldsymbol{\phi}_{k+1}$ is unitary.

The algorithm iterates until the change in \(\mathbf{\Phi}\) is below a predefined threshold \(\epsilon\):
\begin{align}
\|\mathbf{\Phi}_{k+1} - \mathbf{\Phi}_k\|_F < \epsilon.
\end{align}
where \(\|\cdot\|_F\) is the Frobenius norm, which measures the difference between \(\mathbf{\Phi}_{k+1}\) and \(\mathbf{\Phi}_k\). If the change is small, the algorithm has converged to a stationary point. Once convergence is achieved, the algorithm outputs the optimized phase shift matrix \(\mathbf{\Phi}^*\).

\begin{figure}[!t]
\centering
\includegraphics [width=.42\textwidth]{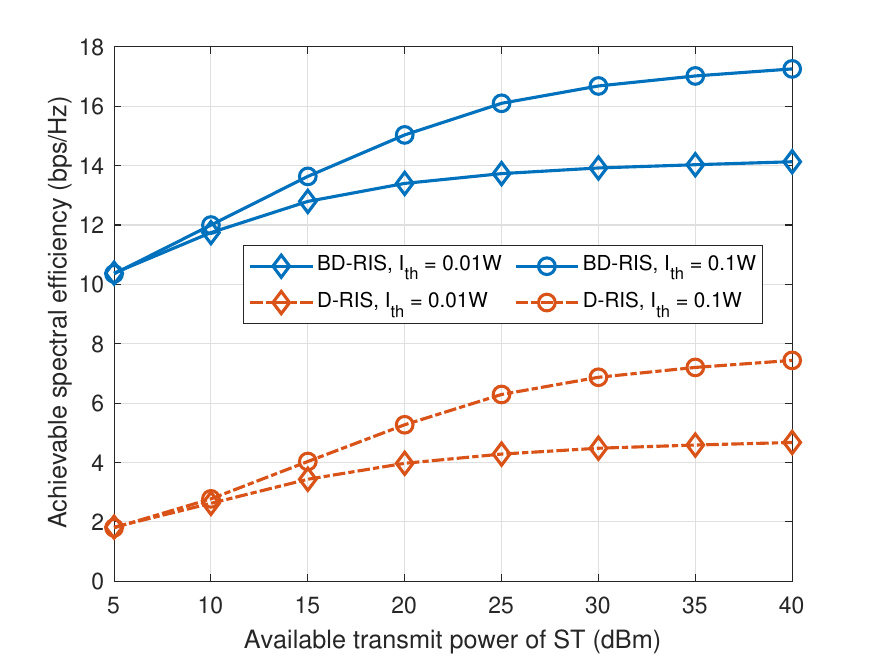}
\caption{Varying allocated power of ST versus achievable spectral efficiency of secondary network, where $M=32$ and $Q_{p}=50$ dBm.}
\label{Power1}
\end{figure}
\section{Numerical Results}  \label{sez_4}
This section provides the numerical results of the proposed framework. We compare the proposed BD-RIS mounted secondary HAPS network with a benchmark D-RIS mounted framework. Unless discussed otherwise, the simulation parameters are set as: the transmit power of HAPS is $P_s=30$ dBm, the transmit power of primary LEO is $Q_p=40$ dBm, the number of reconfigurable elements is $M=32$, Rician factor is $K=10$, and the number of Monte Carlo trials is 1000.

Fig. \ref{Power1} illustrates the achievable spectral efficiency of the secondary network versus the available transmit power of ST under two interference temperature thresholds ($I_{\text{th}}$) for BD-RIS and D-RIS configuration. For $I_{\text{th}} = 0.01$W, the performance is constrained as the ST quickly reaches the interference limit, resulting in a nearly flat spectral efficiency curve at higher transmit power levels. In contrast, for $I_{\text{th}} = 0.1$W, the higher threshold allows the ST to transmit with greater power, leading to significantly improved spectral efficiency. Additionally, BD-RIS outperforms D-RIS across all scenarios due to its interconnected elements, highlighting its superiority in maximizing spectral efficiency while adhering to interference constraints.

Fig. \ref{RIS1} shows the achievable spectral efficiency versus the number of reconfigurable elements in BD-RIS and D-RIS systems under two interference temperature thresholds ($I_{\text{th}}$). For both $I_{\text{th}} = 0.01$W and $I_{\text{th}} = 0.1$W, increasing the number of elements enhances the spectral efficiency due to improved beamforming and signal focusing capabilities. However, the performance gap between the two thresholds persists, with $I_{\text{th}} = 0.1$W enabling significantly higher efficiency because the ST can operate at higher transmit power without violating the interference constraint. Moreover, BD-RIS consistently outperforms D-RIS across all element values, demonstrating the advantage of BD-RIS in effectively leveraging the non-diagonal entries in the phase shift matrix to enhance spectral efficiency and manage interference, especially under high interference temperature thresholds.
\begin{figure}[!t]
\centering
\includegraphics [width=.42\textwidth]{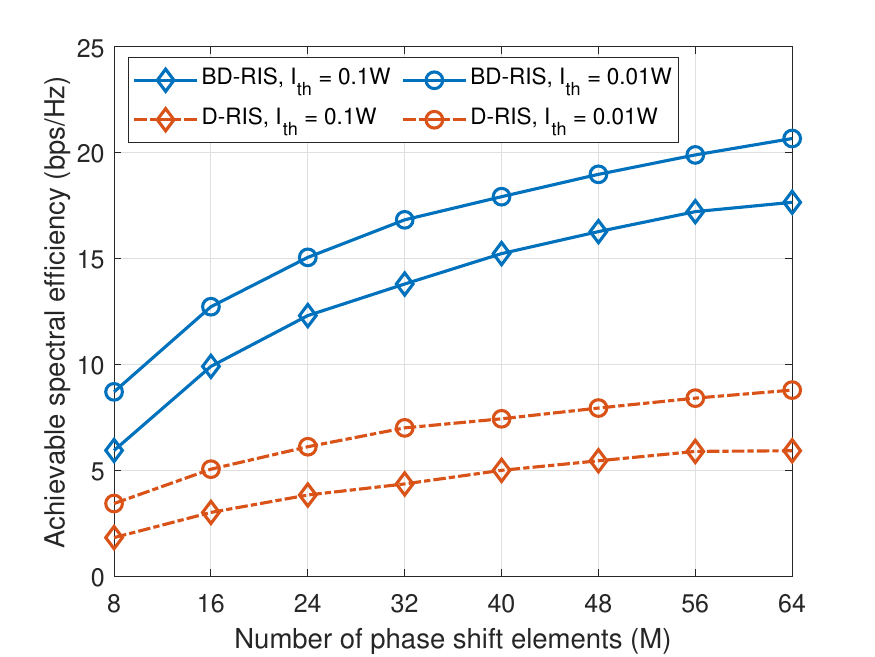}
\caption{Vaying number of BD-RIS elements versus achievable spectral efficiency of secondary network, where $P_s=30$ and $Q_{p}=40$ dBm.}
\label{RIS1}
\end{figure}
\begin{figure}[!t]
\centering
\includegraphics [width=.24\textwidth]{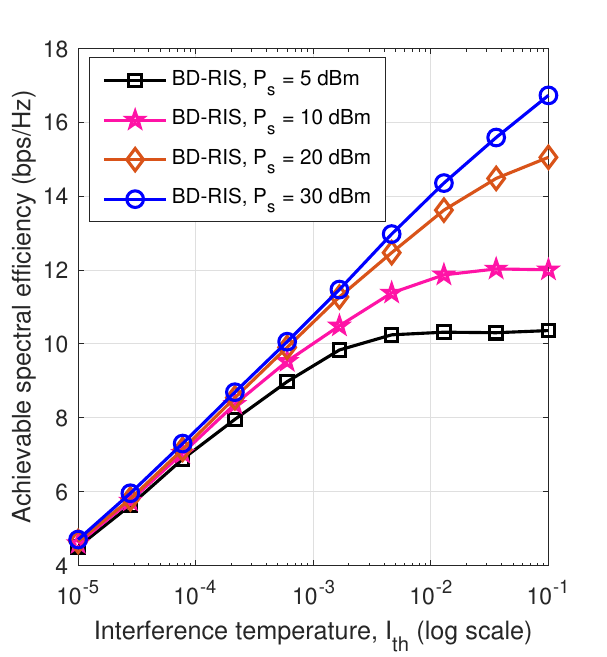}
\includegraphics [width=.24\textwidth]{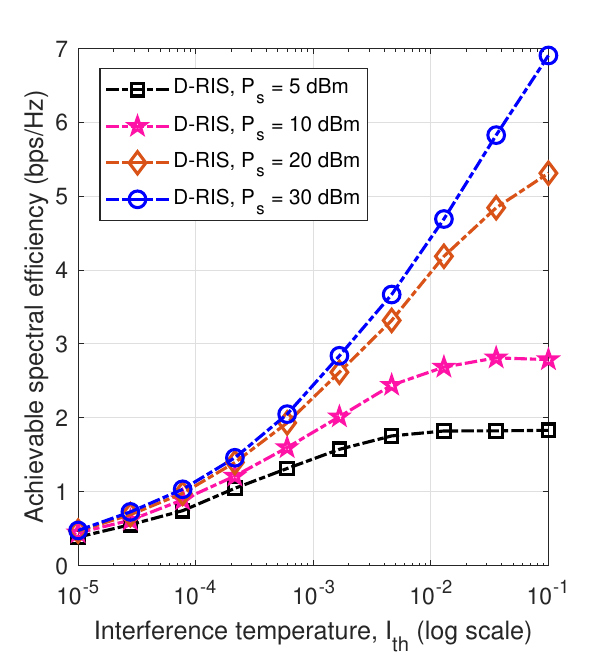}
\caption{Achievable spectral efficiency of secondary network versus the varying $I_{th}$, for the BD-RIS and D-RIS, where $M=32$ and $Q_{p}=40$ dBm.}
\label{Ith}
\end{figure}

Fig. \ref{Ith} shows the achievable spectral efficiency of the secondary network as a function of the interference temperature threshold $I_{\text{th}}$ in a BD-RIS-assisted system, with varying transmit power $P_s$ of the ST. As $I_{\text{th}}$ increases, the spectral efficiency improves because the ST can transmit at higher power levels without violating interference constraints. However, for $P_s=5$ dBm, the spectral efficiency initially rises with $I_{\text{th}}$ but plateaus at higher thresholds. This saturation occurs because the ST’s limited power budget restricts further improvement even when $I_{\text{th}}$ allows for higher transmission. In contrast, for higher $P_s$ values (e.g., 20 dBm or 30 dBm), the spectral efficiency increases significantly with increasing values of $I_{\text{th}}$, demonstrating the combined impact of a higher power budget and a relaxed interference limit. This trend highlights the crucial interplay between $I_{\text{th}}$ and $P_s$ in achieving optimal spectral efficiency.

\section{Conclusion}  \label{sez_5}
This paper introduced a novel framework for cognitive radio-based multilayer NTNs, where BD-RIS mounted HAPS operates as a secondary network underlying a primary LEO network. The main objective was to optimize power allocation for secondary transmitters and phase shift design for BD-RIS in the secondary network while ensuring the quality of service for the primary network. The non-convex problem of maximizing achievable spectral efficiency was handled using an alternating optimization approach. Specifically, the joint optimization problem was decomposed into two subproblems: power allocation and phase shift design. The power allocation problem was solved using the water-filling algorithm and KKT conditions, while the phase shift design was tackled through Riemannian optimization. Simulation results demonstrate that the proposed cognitive radio-based multilayer NTNs with BD-RIS significantly enhance system performance and outperform conventional systems with diagonal RIS configurations.

\section*{Appendix A: Derivation of Optimal Power Allocation Using KKT Conditions}

To obtain the optimal transmit power \( P_s^* \), we first write Lagrangian of our objective function as

\begin{equation}
    \mathcal{L}(P_s, \lambda, \mu) = \log_2 \left( 1 + \gamma_s \right) 
    - \lambda \left( |\mathbf{g} \mathbf{\Phi}|^2 P_{s} - I_{\text{th}} \right)
    - \mu P_s.
\end{equation}
Now we employ Karush-Kuhn-Tucker (KKT) conditions for optimality where stationary condition can be written as
   \begin{equation}
       \frac{d\mathcal{L}}{dP_s} = \frac{|\mathbf{h} \mathbf{\Phi}|^2}{\ln(2) \left( |\mathbf{h} \mathbf{\Phi}|^2 P_s + \sigma^2+|f|^2 Q_p \right)} - \lambda |\mathbf{g} \mathbf{\Phi}|^2 - \mu = 0.
   \end{equation}
and the complementary slackness conditions are given as:
   \begin{equation}
       \lambda \left( |\mathbf{g} \mathbf{\Phi}|^2 P_s - I_{\text{th}} \right) = 0, \quad \mu P_s = 0.
   \end{equation}
The primal feasibility can be given as:
   \begin{equation}
       0 \leq P_s \leq P_{\max}, \quad |\mathbf{g} \mathbf{\Phi}|^2 P_s \leq I_{\text{th}}, \quad \lambda \geq 0, \quad \mu \geq 0.
   \end{equation}
Solving for \( P_s^* \), we Rearrange the stationarity condition as:
\begin{equation}
    P_s = \frac{1}{\lambda \ln(2) |\mathbf{g} \mathbf{\Phi}|^2} - \frac{\sigma^2+|f|^2 Q_p}{|\mathbf{h} \mathbf{\Phi}|^2}.
\end{equation}
Subsequently, we use the complementary slackness condition and solve for \( \lambda \) as:
\begin{equation}
    \lambda = \frac{|\mathbf{g} \mathbf{\Phi}|^2}{I_{\text{th}}}.
\end{equation}
Substituting \( \lambda \) back, we obtain the water-filling solution as:

\begin{equation}
    P_s^* = \min \left( \left[ \frac{1}{\lambda} - \frac{\sigma^2+|f|^2 Q_p}{|\mathbf{h} \mathbf{\Phi}|^2} \right]^+, P_{\max} \right).
\end{equation}
where \( [x]^+ = \max(0, x) \) ensures non-negativity. Thus, the optimal power follows a water-filling strategy, where power is allocated based on the interference constraint and available channel gain. The value of \( \lambda \) is determined by ensuring that the interference power at the PU does not exceed \( I_{\text{th}} \).

\bibliographystyle{IEEEtran}
\bibliography{Wali_EE}

\end{document}